%% file: main.tex
\documentclass[letterpaper, 10pt,conference,twocolumn]{ieeeconf}

\IEEEoverridecommandlockouts  
\input{packages}
\input{newcommands}

\begin{document}
\title{Encirclement Guaranteed Cooperative Pursuit with\\ Robust Model Predictive Control}
\author{Chen Wang$^{1,2}$, Hua Chen${^1}$, Jia Pan${^2}$, and Wei Zhang${^{1,3}}$ 
\thanks{$^1$Department of Mechanical and Energy Engineering, Southern University of Science and Technology, Shenzhen, China. Emails: {\tt chenh6@sustech.edu.cn, zhangw3@sustech.edu.cn}}
\thanks{$^2$Department of Computer Science, The University of Hong Kong. Email: {\tt cwang5@cs.hku.hk, jpan@cs.hku.hk}}
\thanks{$^3$Peng Cheng Laboratory, Shenzhen, China}
\thanks{Corresponding author: Wei Zhang and Jia Pan}
\thanks{This work was supported in part by National Natural Science Foundation of China under Grant No. 62073159 and Grant No. 62003155, in part by the Shenzhen Science and Technology Program under Grant No. JCYJ20200109141601708, and in part by the Science, Technology and Innovation Commission of Shenzhen Municipality under grant no. ZDSYS20200811143601004.}
}
\maketitle

\begin{abstract}
This paper studies a novel encirclement guaranteed cooperative pursuit problem involving $N$ pursuers and a single evader in an unbounded two-dimensional game domain. Throughout the game, the pursuers are required to maintain encirclement of the evader, i.e., the evader should always stay inside the convex hull generated by all the pursuers, in addition to achieving the classical capture condition. To tackle this challenging cooperative pursuit problem, a robust model predictive control (RMPC) based formulation framework is first introduced, which simultaneously accounts for the encirclement and capture requirements under the assumption that the evader's action is unavailable to all pursuers. Despite the reformulation, the resulting RMPC problem involves a bilinear constraint due to the encirclement requirement. To further handle such a bilinear constraint, a novel encirclement guaranteed partitioning scheme is devised that simplifies the original bilinear RMPC problem to a number of linear tube MPC (TMPC) problems solvable in a decentralized manner. Simulation experiments demonstrate the effectiveness of the proposed solution framework. Furthermore, comparisons with existing approaches show that the explicit consideration of the encirclement condition significantly improves the chance of successful capture of the evader in various scenarios. 
\end{abstract}

\section{Introduction}\label{sec:intro}
Pursuit evasion, as one of the most important problems in the multi-agent literature, mainly concerns with maneuvering a group of agents called the pursuers to capture one or more evaders acting adversarially against the pursuers. Such a problem has been attracting a great amount of research attentions during the past several decades, due to its broad applicability in various practical scenarios including~\cite{Issac1965, littlewood1986, NOWAKOWSKI1983235, AIGNER19841, chung2011search}.


The pursuit evasion problem, albeit intuitively straightforward to envision, is fundamentally challenging to solve. Introduced as the ``lion and man'' problem, early-stage investigations on pursuit evasion problems mainly focused on establishing conditions that ensure existence of capturing strategies for the lions (pursuers)~\cite{Croft1964, littlewood1986, V1978, ibragimov1998game}. Under the continuous-time setting, with additional assumptions that the game domain is unbounded and the maximal speeds for the lions and the man are the same, Jankovic~\cite{V1978} showed that capturing is achievable if and only if the man is always kept inside the encirclement of the lions. Later works extend this result to discrete time setting~\cite{alexander2009capture, kopparty2005framework}. However, a pitfall in the adopted settings is that the evader's action is assumed to be known to the pursuers which gives one step advantage to the pursuers. More recently, generalizations of these classical results to randomized scenarios~\cite{Isler2005}, complex domains~\cite{Deepak2012} and different capturing conditions~\cite{Lavalle2001, Gerkey2006} have been continuously reported. Chen {\em et al.}~\cite{chen2016multi} explicitly considers the encirclement condition in the problem, reporting a minimum of $6$ pursuers for encirclement guarantees. 



On the other hand, the game-theoretic perspective on the pursuit evasion problem has been widely employed. Such a viewpoint leads to a differential game based reformulation of the pursuit evasion problem~\cite{Issac1965} whose optimal solution is fully characterized by the renowned Hamilton-Jacobi-Issacs partial differential equation (HJI PDE)~\cite{bacsar1998dynamic, Issac1965, captureFlag}. Nonetheless, the HJI PDE suffers from the curse of dimensionality. To alleviate the computational complexity, various strategies have been developed. Voronoi partition based method has been proposed initially for bounded convex space \cite{huang2011guaranteed}.  Later on, this method has been extended in several directions.  Zhou {\em et al.}~\cite{zhou2016cooperative} extends the concept of Voronoi partition to safe reachable set that is applicable to non-convex game domains.  Pierson {\em et al.}~\cite{Mac17RAL} extends the Voronoi partition based approach to high dimensional game domains and scenarios where the evader's position is uncertain~\cite{SHAH2019103246}.

In this paper, we study cooperative encirclement guaranteed pursuit problem involving multiple pursuers and a single evader in an unbounded 2D game domain under the assumption that the evader's action is fully unknown to all the pursuers. With this restriction on the underlying information structure, we explicitly take the encirclement condition into account in the problem formulation, in addition to the capturing condition. It is worth noting that, despite the fact that the importance of encirclement condition has been revealed decades ago in~\cite{V1978}, such a condition has been rarely exploited in constructing solutions to cooperative pursuit problems. To address the difficulties introduced by the limited information structure and the explicit consideration of the encirclement condition, a robust model predictive control (RMPC) framework is developed. Such a RMPC framework views the unknown evader's velocity as disturbance to the relative dynamics between the pursuers and the evader, and tries to find the pursuers' optimal actions ensuring the encirclement condition regardless of how the evader behaves. Unfortunately, the RMPC reformulation remains a bilinear optimization for the underlying problem. To tackle this bilinear issue, a novel encirclement guaranteed partitioning scheme is devised that simplifies the original RMPC problem to a number of tractable linear tube model predictive control (TMPC) problems. Numerical experiments demonstrate effectiveness of the proposed solution framework, whereas existing approaches fail to ensure encirclement and/or capture. 


The main contributions of this paper are summarized below. First, we consider a novel and practically important cooperative pursuit problem, whose task involves both capturing the evader as well as ensuring the encirclement condition. Second, a robust model predictive control framework is established to solve the underlying problem while assuming unavailability of the evader's action, which is generalizable to other scenarios with imperfect or restrictive information structures. Third, to efficiently solve the reformulated RMPC problem, a novel encirclement guaranteed partitioning scheme is devised, which fully exploits the features of the underlying problem. With the help of this partitioning scheme, the original bilinear RMPC problem reduces to a number of linear tube MPC problems that are practically solvable in a decentralized manner.

\section{Problem Description}\label{sec:prob_form}
In this paper, we are concerned with a cooperative pursuit problem involving $N$ pursuers and $1$ evader in a two-dimensional game domain. Generally speaking, the goal of the pursuers is to restrict the evader's possible motion and eventually capture the evader, while the evader aims to escape the pursuers' surrounding and avoid being captured. 

Denoting by $\p^i \in \R^2$ the position of the $i$-th pursuer and $\e\in \R^2$ the position of the evader, the discrete-time dynamics of these two classes of agents are given by 
\subeq{\al{ \e(t+1) &= \e(t) + \mathbf{u}_e(t), & \e(0) &= \e_0\\  \p^i(t+1) &= \p^i(t) + \mathbf{u}_p^i(t), & \p^i(0) &= \p_{0}^i,  \forall i\in \NN_{N}}} where $\mathbf{u}_e(t)\in \R^2$ is the evader's control input, $\mathbf{u}_{p}^i$ is the $i$-th pursuer's control input, $\e_0\in \R^2$ and $\p_{0}^i\in \R^2$ are the initial positions for the evader and $i$-th pursuer, respectively. We also use $\mathbb{N}_{N}$ to denote the number set $\{0,1,\dots,N-1\}$. To approximately account for actuation limits that commonly arise in practice, the following input constraints are imposed. 
\subeq{\al{\mathbf{u}_e \in \W &\triangleq \{\mathbf{u}_e| \norm{\mathbf{u}_e}_\infty \leq u_{e, max}\} \\  \mathbf{u}_{p}^i \in \U &\triangleq \{\mathbf{u}_p | \norm{\mathbf{u}_p}_\infty \leq u_{p, max}\}}}



To mathematically characterize the pursuers' first objective of restricting the evader's motion, we employ the following encirclement condition throughout this paper. 


\begin{definition}[Encirclement Condition]
The evader is said to be encircled by the pursuers, if it lies inside the convex hull $\E$ formed by all the pursuers, i.e., 
\aln{\e \in \E \triangleq \{\p\in \R^2: &\p =   \sum\limits_{i=0}^{N-1} \lambda^i \p^i, \lambda^i\ge 0 , \sum\limits_{i=0}^{N-1} \lambda^i= 1 \}}
\end{definition}

Fig.~\ref{fig:quadrant}(a) shows the encirclement condition. It is obvious that if the above condition is ensured indefinitely then the evader cannot escape from the region formed by all the pursuers. 

Having the encirclement condition satisfied, the region where the evader is allowed to freely move may still remain large. Hence, at least one of the pursuers needs to move close enough to the evader to ensure capture of the evader, which formalizes the following capture condition. 

\begin{definition}[Capture Condition]
The evader is said to be captured if there exists at least one pursuer whose distance from the evader is smaller than the capture radius $r_c$, i.e.,
\eqn{\exists i\in \NN_{N}, \text{ such that } \norm{\p^i(t) - \e(t)}_2 \leq r_c}
\end{definition}

Throughout this paper, we adopt the information structure that the evader's state (i.e., position) is available to all the pursuers, while its control input (i.e., velocity) remains private to itself. Under this information structure, we are concerned with the encirclement guaranteed cooperative pursuit problem which aims to find control inputs for the pursuers to persistently ensure encirclement and eventual capture the evader regardless of the evader's behavior.




\begin{problem}[Encirclement Guaranteed Cooperative Pursuit]\label{prob:eg_cp}
Given an initialization of the pursuers and the evader satisfying the encirclement condition, find the control law for the pursuers (either centralized or decentralized) such that encirclement condition is always satisfied and capture condition is achieved eventually.
\end{problem}


To tackle the above problem, the key challenge lies in the unavailability of the evader's control information to the pursuers. As a consequence, the pursuers' control inputs need to take into account of the worst case scenario in which the evader is always choosing its best possible control.
\begin{remark}\label{rem:game_sln}
The differential game approach is often adopted to address this difficulty. However, formulation and solution approaches following such a game-theoretic perspective are typically challenging to develop. More importantly, such game-theoretic based approaches suffer from the curse-of-dimensionality and is not numerically tractable in general. 
\end{remark}

In this paper, a novel robust model predictive control based perspective is employed, which leads to a tractable and decentralized solution. In the sequel, the robust model predictive framework for the encirclement guaranteed cooperative pursuit problem is developed.


\section{Overview of the Proposed Robust Model Predictive Control Framework}

Model predictive control (MPC), as one of the most powerful optimization-based approaches in robotics applications, iteratively solves an optimal control problem in real time and only applies the control input at the first time step to the system. Robust model predictive control (RMPC), in principle, adjusts the conventional MPC formulation to account for disturbances or uncertainties. 


At a generic time $t$, the goal of the pursuers can be formulated as minimizing the accumulated distances between the pursuers and the evader over a look-ahead horizon while ensuring encirclement regardless of the evader's input. Let $k$ be the prediction time index, and let $\p_k^i$ and $\e_k$ be the position of pursuer $i$ and the evader at time $t + k$. The evader's input can be viewed as a disturbance to the system, and thus we denote $\w_k  = - \mathbf{u}_e(t+k)$. At time $t$, the pursuit control can be obtained by solving the following RMPC problem.
\subeq{\label{eq:rmpc}\al{ \min\limits_{ \mathbf{u}_k^i \in \U, \lambda_k^i,\forall k}& \    \sum_{i=0}^{N-1} \sum_{k=1}^{T} \norm{\p_{k}^i - \e_{k}}^2_2 \\  \st \ \  & \p_{k+1}^i = \p_k^i + \mathbf{u}_k^i,   \forall k, \forall i \label{eq:dync1}\\ & \e_{k+1} = \e_k - \w_k, \forall k \label{eq:dync2}\\  & \label{eq:biCons} \sum_{i=0}^{N-1} \lambda_k^i\p_k^i = \e_k,\forall k\\  & \sum_{i=0}^{N-1} \lambda_k^i = 1,\forall k; \  \lambda_k^i \geq 0, \forall k, \forall i \label{eq:biCons3}}}
The constraints should hold for all possible $\mathbf{w}_k \in \W$. The above problem should be re-solved at each time $t$ based on updated state observations. Here, the cost function accounts for the capture condition via the accumulated distance between the evader and all the pursuers.  The encirclement condition is enforced via~\eqref{eq:biCons} and~\eqref{eq:biCons3}.  For any initial condition, if the RMPC problem is feasible, then no matter how the evader moves, there always exists controls for all the pursuers such that the evader will remain inside the encirclement of the pursuers and the cost function is minimized.

Note that the RMPC based reformulation does not change the zero-sum game nature of the problem, and the resulting RMPC problem remains challenging to solve, particularly due to the bilinear constraints guaranteeing encirclement~\eqref{eq:biCons} and~\eqref{eq:biCons3}. To address the bilinear constraint and efficiently solve the RMPC problem, we devise a tractable MPC based solution scheme. Roughly speaking, the core idea lies in approximating the bilinear encirclement constraint through a set of linear constraints that are carefully constructed via exploiting the structure of the underlying problem. This approximation technique subsequently leads to a simple RMPC problem that can be efficiently solved via quadratic programming based methods in a decentralized manner.

In more details, the key observation inspiring our approach lies in that the encirclement condition is persistently satisfied if there always exists a pursuer that stays inside each quadrant of the frame centered at the evader, as shown in Fig.~\ref{fig:quadrant}(b). Motivated by this observation, the proposed solution framework involves two main steps. First, a special partition of the two-dimensional game domain is constructed, which ensures the encirclement condition if each partition set always contains at least one pursuer during the pursuit process. Then, the cooperative controller synthesis problem for the pursuers can be recast as the problem of controlling each pursuer to stay inside the correspondingly assigned partition set. Such a specific problem can subsequently be reformulated as a number of simple tube model predictive control (TMPC) problem separable among all pursuers, which only needs to be solved in a decentralized manner. 



\begin{figure}[tp!]
	\centering
	{\vspace{0.1in}
	\includegraphics[width=1.0\linewidth]{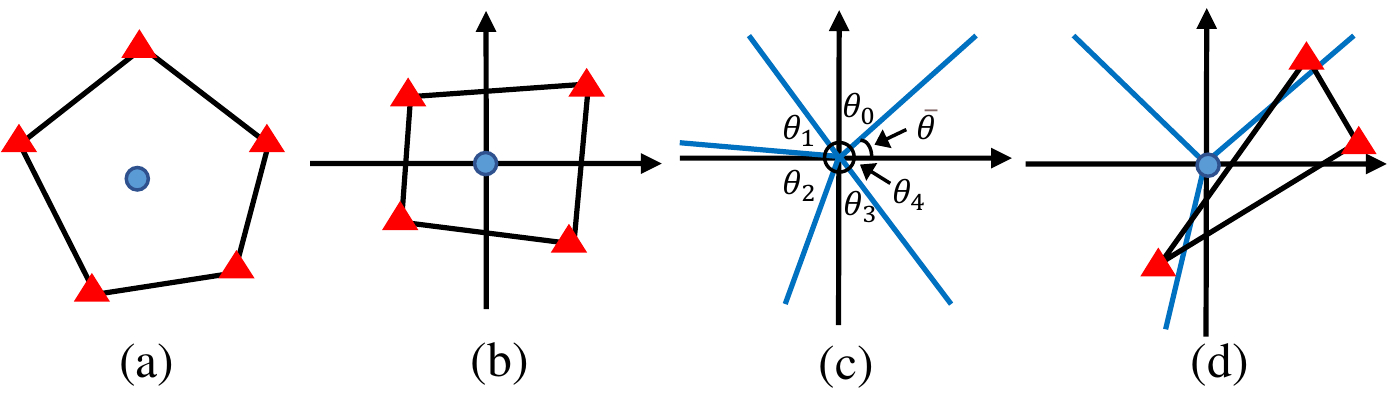}}
	\caption{\footnotesize Red triangles represent pursuers and blue circle represents evader. Blue line represents angle-based Partition. (a) Encirclement condition. (b) Satisfaction of the encirclement condition with four pursuers. (c) Angle based partition. (d) Example of angle-based partition that is not an Encirclement Guaranteed Partition.}
	\label{fig:quadrant}
	\vspace{-10px}
\end{figure}
\section{Encirclement Guaranteed Partition based Cooperative Pursuit via Tube MPC}

\subsection{Encirclement Guaranteed Partition (EGP)}\label{sec:egp}
As briefly mentioned in the previous section, the first step of the proposed solution framework involves construction of a carefully designed partition of the two-dimensional domain that inherently ensures encirclement. To begin with, a partition $\P = \{P_m\}$ of the considered game domain $\R^2$ is a set of subsets of $\R^2$, i.e., $P_m\subset \R^2$ such that:
\begin{enumerate}
    \item The union of all $P_m$'s is the overall game domain, i.e., $\bigcup\limits_m P_m = \R^2$,
    \item For any two elements in the partition, their intersection is empty, i.e., $P_m\bigcap P_l = \emptyset$ for all $m\neq l$.
\end{enumerate}

Without loss of generality, consider a reference frame centered at the evader now, as shown in Fig.~\ref{fig:quadrant}(c). By resorting to this evader-centered frame, partition of the game domain can be simply represented via a set of angles, according to the following definition. 
\begin{definition}[Angle-based Partition]\label{def:partition_angle}
The set of angles $\Theta =\left\{\bar{\theta},  \theta_0,\ldots,\theta_{M-1}\right\}$ defines a valid partition for $\R^2$, if 
\eq{\sum\limits_{m=0}^{M-1}\theta_m = 2\pi, \text{ and }  \theta_m>0,  \forall m. } The corresponding partition of $\R^2$ is denoted by $\P^\Theta$.
\end{definition}



In the above definition, $\bar{\theta}$ is used to identify the starting angle of the partition.
It is worth mentioning that, albeit the above definition defines valid partitions of the two-dimensional game domain, the partitions we are looking for need to possess additional properties related to the encirclement requirement. In particular, the following encirclement guaranteed partition is adopted in this section. 

\begin{definition}[Encirclement Guaranteed Partition]\label{def:feas_part}
A valid angle-based partition $\P^\Theta = \left\{P^\Theta_m\right\}_{m=0}^{M-1}$ is an encirclement guaranteed partition, if the convex hull of any $M$ points picked distinctly from the partition elements contains the origin of the frame, i.e., 
\eqn{\ald{&\forall p_m\in P^\Theta_m, \ \forall m \in \NN_{M} \\ &\exists \{\eta_m \geq 0\}_{m=0}^{M-1},    \sum\limits_{m=0}^{M-1} \eta_m = 1, \text{ such that }\0 = \sum\limits_{m=0}^{M-1} \eta_m p_m, }}
\end{definition}

\begin{remark}\label{rem:eg_part}
Apparently, if $\Theta$ defines an encirclement guaranteed partition, then it suffices to keep at least one pursuer inside each partition set to ensure encirclement of the evader. 
\end{remark}

It turns out that encirclement can be guaranteed by a straightforward condition on the angles defining the partition, which is summarized in the following theorem. 



\begin{theorem}\label{the:erp_cond}
An angle-based partition $\P^\Theta$ guarantees encirclement if and only if 
\eq{\label{eq:angle_part_thm}\theta_0 + \theta_1 \leq \pi, \ldots \theta_{m} + \theta_{m+1} \leq \pi, \ldots, \theta_{M-1} + \theta_0 \leq \pi.}
\end{theorem}

 \begin{proof} 
\textbf{Sufficiency:} Suppose the condition holds, then $P_m\cup P_{m+1}$ is a  convex hull. It follows that we can choose any point $x_m$ from $P_m$ and any point $x_{m+1}$ from $P_{m+1}$, such that the line connecting $x_m$ and $x_{m+1}$ falls inside $P_m\cup P_{m+1}$. It follows that the polygon formed by connecting all $M$ points which are picked distinctly from the partitions will contain the origin. This polygon is also contained in the convex hull of all $M$ points. Therefore, the origin is contained in the convex hull and $\P^\Theta$ is an Encirclement Guaranteed Partition.

\textbf{Necessity:} Suppose $\theta_m + \theta_{m+1} > \pi$, then $P_m \cup P_{m+1}$ is non convex. It follows that we can choose one point $x_m$ from $P_m$ and one point $x_{m+1}$ from $P_{m+1}$, such that line $x_mx_{m+1}$ contain line segment outside $P_m \cup P_{m+1}$. It is clear then that no convex combination  of $x_m$, $x_{m+1}$ and other points that are not in $P_m\cup P_{m+1}$ can be origin. 

\end{proof}
Given an angle-based partition, checking whether it guarantees the encirclement property is simply checking~\eqref{eq:angle_part_thm}. Moreover, the following corollary can be concluded. 

\begin{corollary}\label{coro:minPart}
An encirclement guaranteed partition contains at least $4$ elements. 
\end{corollary}



Recall that, the two-dimensional game domain should be partitioned so that the resulting partition is encirclement guaranteed. More importantly, there always exists at least one pursuer staying indefinitely in each corresponding partition set. Therefore, constructing the encirclement guaranteed partition relies on the pursuers' and the evader's initializations. As a result, given any initialization, we need to further investigate whether there exists an associated encirclement guaranteed partition and if so how to find such partition. 

According to Corollary~\ref{coro:minPart}, it is required to have at least $4$ pursuers to encircle a single evader. Without loss of generality, we consider the scenario where $N\geq4$ pursuers are involved and $4\leq M \leq N$ partitions are to be identified. Let the relative angle between each pursuer $i$ and the evader with respect to the evader-centered frame be $\alpha_i$, finding the encirclement guaranteed partition is basically trying to separate the pursuers while satisfying the condition in Theorem~\ref{the:erp_cond}, which is formulated as the following quadratic program (QP).

 \subeq{
 \label{eq:ERPQP}
 \al{ \min \limits_{\bar{\theta}, \theta_0, ..., \theta_{M-1}} \quad &\sum_{i=0}^{M-1} (\frac{2\pi}{M} - \theta_i)^2 + \bar{\theta}^2 \\ 
      \st\quad \label{eq:part_q1}\quad &\alpha_M- 2\pi < \bar{\theta} < \alpha_1  \\
      &  \label{eq:part_q5}\alpha_{j}\!<\! \bar{\theta}\! +\! \sum_{m=0}^{j} \theta_{m} \! <\! \alpha_{j+1}, \forall j \in \NN_{M-1} \\
      & \label{eq:part_q6}\theta_m > 0, \forall m \in \NN_{M} \\
      & \theta_m + \theta_{m+1} \leq \pi, \forall m \in \NN_{M-1} \\
      & \theta_0 + \theta_{M-1} \leq \pi,  \label{eq:part_q9} \sum_{i=0}^{M-1} \theta_i = 2\pi  
     }}

In the above problem, the constraints~\eqref{eq:part_q1} and~\eqref{eq:part_q5} basically requires that the resulting partition separates the pursuers and embraces at least one pursuer in each partition set, while the constraints~\eqref{eq:part_q6} to~\eqref{eq:part_q9} enforces that the resulting partition is encirclement guaranteed. 

To conclude this subsection, the proposed partition scheme is summarized with an algorithm whose pseudo-code is provided in Algorithm~\ref{alg:erp}. 

\begin{algorithm}[bp!]
\SetAlgoLined
Compute ${\alpha_i}, \forall i \in \NN_{N}$\;
From $\alpha_i$ choose a subset with $M$ angles\;
Solve QP (\ref{eq:ERPQP})\;
return solution\;
\caption{EGP construction}\label{alg:erp}
\end{algorithm}

\begin{remark}
In the second step of this algorithm, we need to choose $M$ angles from $N$ angles. Here we have one requirement that the corresponding $M$ pursuers of these $M$ angles form an encirclement of the evader. What's more, we would prefer these $M$ angles to be more evenly spaced so that the returned EGP consists of more equally sized partitions.
\end{remark}

\subsection{Tube MPC-Based Cooperative Pursuit}
Given the encirclement guaranteed partition for any feasible initialization of the pursuers and evader, the cooperative pursuit problem we aim to solve can then be reformulated as the problem of controlling all pursuers to persistently achieve the encirclement constraint. Thanks to the partitioning scheme and its properties, the problem can further be simplified to how to keep each pursuer staying in its assigned partition set, which will be discussed in details here.

By virtue of the angle-based partition strategy, each partition set is polyhedral and hence can be represented via linear inequalities. This key feature, combined with its inherent encirclement preserving property, allows for the desirable linear approximation to the original bilinear RMPC problem~\eqref{eq:rmpc}. To facilitate our controller synthesis, we pick an arbitrary partition $P_m^\Theta\in \P^\Theta$, and suppose the $i_mth$ pursuer is inside $P_m^\Theta$ at time $t$. For notation simplicity, we drop the player index and write the relative dynamics between the $i_mth$ pursuer and the evader as follows:
\eq{\x(t+1) = \x(t) + \mathbf{u}(t) + \w(t) \quad (\text{in Partition } P_m^\Theta) \label{eq:rt_system}}
where $\x = \p^{i_m} - \e$, $\mathbf{u} = \mathbf{u}_p^{i_m}$ and $\w = - \mathbf{u}_e$. Note that, this dynamics basically describes the pursuer's dynamics in the evader-centered frame that was used in the construction of the partition scheme. 


With this reformulated dynamics, solving the following model predictive control problem yields a control input for each pursuer that is capable of ensuring encirclement and improving capture. With $k$ being the prediction time step, the decentralized robust MPC to be solved at time step~$t$ is:
\subeq{\label{eq:dcrmpc}\al{ \min\limits_{\mathbf{u}_k}  \quad  & \sum \limits_{k=0}^{K-1}\norm{\x_k}_2 \\ \st \quad & \x_{k+1} = \x_k + \mathbf{u}_k + \w_k, \forall k\\ 
& \x_0=\x(t)\\ &\x_k \in P^\Theta_m \triangleq \X,  \mathbf{u}_k \in \U,\forall k}}
The constraints should hold for all possible $\mathbf{w}_k \in \W$. Here $\x_k$ and $\mathbf{u}_k$ are the predicted states and controls of the relative dynamics in partition $P_m^\Theta$ given by~(\ref{eq:rt_system}).
In the above MPC problem, the evader's control is viewed as disturbance to the relative dynamics. Therefore, the pursuer's control needs to be able to maintain the pursuer's position inside its corresponding partition for any possible evader's control. Thanks to the polyhedral nature of the constraint sets, the above MPC problem~\eqref{eq:dcrmpc} can be solved by TMPC.

The TMPC problem is typically solved in a two-step fashion~\cite{MAYNE2005219, Sasa2005}. At the first step, a so-called invariant set $\S$ which specifies the effect of the evader's action to the uncertainties of the pursuer's state is computed. Such a set is mathematically the minimal robust positively invariant set of the underlying dynamics, which can be efficiently computed via polyhedral set operations. Once such a set is obtained, it will be utilized as part of the state constraint in a linear MPC problem. Through solving the resulting linear MPC, the desired control for the pursuer can finally be constructed.

In the sequel, more details regarding the formulation and solution of the TMPC problem are provided. Especially, thanks to the simplicity of the relative dynamics~\eqref{eq:rt_system} and the polyhedral nature of the state and input constraints, the environment set $\S$ can be analytically determined. 


\subsubsection{System Decomposition}
\label{sec:sys_decom}
To begin with, we decompose the relative dynamics~\eqref{eq:rt_system} into a nominal part and a difference part.  The nominal part is $\z(t+1) = \z(t) + \mathbf{v}(t)$, where $\z$ is the nominal state, and $\mathbf{v}(t)$ is the nominal control. The difference between $\x(t)$ and $\z(t)$ is refer to as the difference term $\s(t) = \x(t) - \z(t)$. The control law is chosen as $\mathbf{u}(t) = \mathbf{v}(t)  -\s(t)$. It follows that $\s(t+1) = \w(t)$.  Therefore $\s \in \S = \W$ where $\S$ is the Invariant Set in TMPC. Under this decomposition, $\forall \x(t) \in \z(t) \oplus \S, \forall \w(t) \in \W$, we have $\x(t+1) \in \z(t+1)\oplus \S$. Here $\oplus$ is the Minkowski addition.


\subsubsection{Constraints Construction} 

 
We then begin to construct the constraints for the nominal state. The nominal state should satisfy $\z(t)\oplus \S \subseteq \X$ since $\x \in \X$. Equally speaking, $\z(t) \in \overline{\X}$, where $\overline{\X}$ is $\X \ominus \S $ and $\ominus$ is the Pontryagin set difference.  Similarly, $\mathbf{v} \oplus (-\S) \in \U$ since $\mathbf{u} \in \U$. Therefore, $\mathbf{v}\in \overline{\U}$, where $\overline{\U}$ is $\U \ominus (-S)$. Note that $\z \in \overline{\X}$ and $\mathbf{v} \in \overline{\U}$ are the constraints for the nominal system. We use the Multi-Parametric Toolbox 3 to do the computation \cite{MPT3}.



\subsubsection{Quadratic Programming Formulation}
Now we can form the quadratic programming (QP) of TMPC with $K$ horizon.  The input of the quadratic optimization problem is $\x$ which is the current position of the pursuer, and the solution to this problem is a sequence of control laws. Suppose that the current position of pursuer $i$ is $\x$. The optimization to be solved is:
\subeq{\label{eq:TMPCOp}\al{ \Min\limits_{\substack{\z_0,\z_1,...,\z_K\\ \mathbf{v}_0,\mathbf{v}_1,...\mathbf{v}_{K-1}}} & \sum \limits_{k=0}^{K-1}(\norm{\z_k}_Q+\norm{v_k}_R) + \norm{z_K}_P \\  \st \quad & \z_{k+1} = \z_k + \mathbf{v}_k, \mathbf{v}_k \in \overline{\U}, \forall k \in \NN_{K}\\ & \z_k \in \overline{\X}, \forall k \in \NN_{K+1}\\ & \x - \z_0 \in \S  }}
Here $\overline{\X} = \X \ominus \S$, $\overline{\U} = \U \ominus \S$, and $\S = \W$.  Note that the Pontryagin Set Difference of two polytopes is still a polytope. Therefore, this optimization is a quadratic program and can be solved efficiently.

Solution to this optimization problem is sequences of nominal states $\z_k$ and controls $\mathbf{v}_k$. Let $\x_0, \x_1, ..., \x_{K-1}$ with $\x_0 = \x$ be the possible states sequence of (\ref{eq:rt_system}). The resulting control laws is then given by
\eq{\mu_k(\x_k) = \mathbf{v}_k  -\x_k + \z_k, \forall k \in \NN_{K}.\label{eq:rtmpc_ctrl_law}} According to Section \ref{sec:sys_decom}, we have $\x_k \in \z_k \oplus \S$ under control laws (\ref{eq:rtmpc_ctrl_law}). Therefore with this control law, we can ensure $\x$ in the next time step will still stay in $\overline{\X} \oplus \S \subseteq \X$ under any $\w \in \W$. Hence TMPC method has recursive feasibility and by applying the TMPC control law we can ensure the constraints in (\ref{eq:TMPCOp}) is never violated. 
\begin{remark}
In order for $\overline{\U}$ to be non-empty, we need to have $u_{p, max} \geq u_{e, max}$. If $u_{p, max} = u_{e, max}$, $v_k$ is always $0$ and only the ancillary controller $z_k - x_k$ plays the role.
\end{remark}

\begin{algorithm}[htbp]
\SetAlgoLined
$t=0$\;
Calculate $\x_i(t) = \p_i(t) -  \e(t)$,$i\in\NN_{N}$\;
Construct EGP $\P^\Theta$ with 4 partitions\;
Construct QP  $\Q_i$ for each partition element $P_i$ of EGP\;
\While{Evader is not captured}{
\For{each $P_i$}{
Pursuer farthest from the origin in $P_i$ get its control by solving $\Q_i$\;
Other pursuers in $P_i$ use Direct Charge method\;
}
pursuers move\;
$t = t+1$\;
\caption{TMPC-based Pursuit Framework}\label{alg:tmpc} }
\vspace{-0.07in}
\end{algorithm}
\subsection{Discussions and Implementation Details}\label{sec:imp_dis}
To sum up, further discussions on the developed framework and implementation details are provided. 


First, the partition and TMPC based solution framework developed in this section can be viewed as a linear approximation to the bilinear RMPC problem~\eqref{eq:rmpc}, which exploits the special structure of the underlying problem setting. Owing to Corollary~\ref{coro:minPart}, it is immediate that if more than $4$ pursuers are involved in the problem, it suffices to assign $4$ pursuers to ensure the encirclement condition, while the remaining pursuers are simply assigned the task of capturing the evader without caring too much about the encirclement requirement. 

Inspired by the above observation, one simple and effective way of implementing the algorithm involves assigning the task of encirclement guarantee to $4$ selected pursuers, and all others are assigned the task of simply chasing the evader. By doing so, persistent satisfaction of the encirclement condition can be theoretically ensured, while capturing condition can be empirically guaranteed. Pseudo-code of an algorithm summarizing the developed framework is given in Algorithm~\ref{alg:tmpc}.

It should be noted that it is possible to have a decentralized implementation where all pursuers get their control by solving the corresponding TMPC problem. Specifically, communication among the pursuers is only needed at initialization, when partitions and the corresponding assignment need to be determined. Afterwards, each pursuer's control input depends only on its own and the evader's state.

\section{Simulation Results}\label{sec:res}

To validate the effectiveness of the proposed framework, we first verify the  partitioning scheme developed in Section~\ref{sec:egp}. Then, the performance of the overall framework is demonstrated. The reader is kindly referred to the supplementary video for comprehensive simulation results\footnote{\url{https://sites.google.com/view/egcp/}}.

\subsection{Effectiveness of the Proposed EGP Scheme}

Performance of Algorithm~\ref{alg:erp} with a number of randomly initialized configurations of pursuers and evader is demonstrated in Fig.~\ref{fig:ERP}. A total of $8$ scenarios are tested, in which first three columns show the results involving $N = 7$ pursuers randomly initialized with uniform distribution within the box region $[-5,-5]\times[-5,-5]$ while the last column shows scenarios of initializations that do not admit EGPs. From the figure, it can be easily seen that the proposed Algorithm identifies EGPs under proper initializations.

\begin{figure}[htbp]
	\centering
	{\includegraphics[width=1\linewidth]{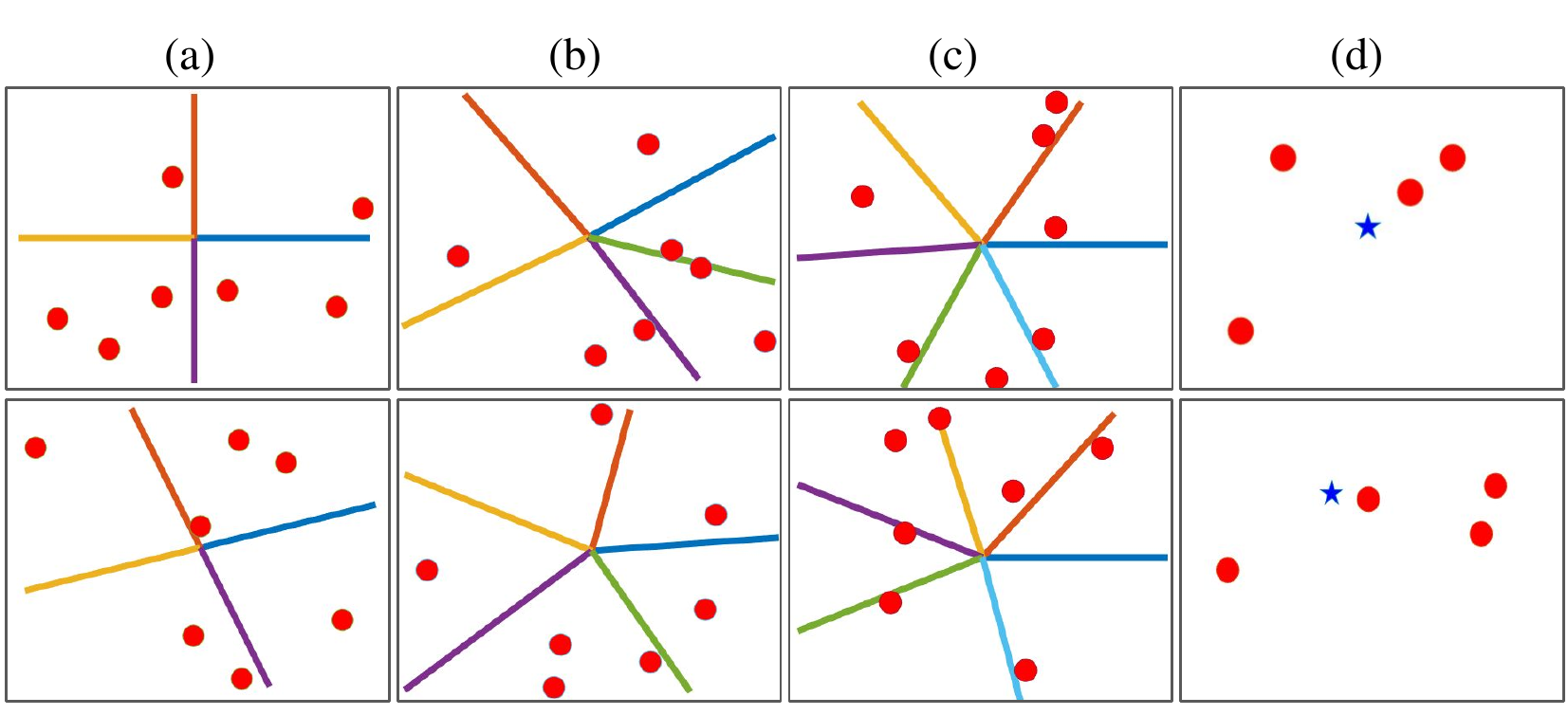}}
	\caption{\footnotesize Red circle represents pursuer. Colored line represents EGP. Blue star represents evader. Column (a): EGP with $7$ randomly generated pursuers and $4$ partitions; (b): EGP with $7$ randomly generated pursuers and $5$ partitions; (c): EGP with $7$ randomly generated pursuers and $6$ partitions; (d): Cases where no EGP exists.  }
	\label{fig:ERP}
	\vspace{-10px}
\end{figure}

\subsection{Performance of the Overall Framework}

With the effective partitioning scheme, performance of the overall proposed framework can then be validated. In the numerical test, $5$ pursuers are considered. Input limits for the pursuers and the evader are set to be $u_{e, max} = 1$ and $u_{p, max} = 1.1$. Capture radius defining the capture condition is set to be $r_c=5$. Parameters in Algorithm~\ref{alg:tmpc} are set to be $Q=I$, $R = \mathbf{0}$ and $P=3I$, where $I$ is the $2\times2$ identity matrix and $\mathbf{0}$ is the $2\times2$ zero matrix.


Our proposed TMPC-based pursuit framework Algorithm~\ref{alg:tmpc} is implemented in MATLAB. The Voronoi Partition based method~\cite{Mac17RAL} and the ``Direct Charge" method are implemented as baselines for comparison. Roughly speaking, the Voronoi Partition method first constructs a Voronoi partition given the positions of all pursuers and evader, then commands each pursuer to move towards the centroid of the common boundary between the evader and the corresponding pursuer. One thing to note is that once the evader is outside the encirclement of the pursuers, the Voronoi partition method will fail because the area of evader's Voronoi partition become infinity. In our implementation of Voronoi Partition based method, we use a virtual boundary to calculate the pursuers' controls once the evader is outside the encirclement. In the ``Direct Charge" method, the pursuer moves directly towards the evader with its maximal velocity.


In the simulations, $5$ pursuers initialized at $(10, 90)$, $(-60, 80)$, $(-90, -90)$, $(90, -10)$ and $(-90, 30)$ participate in the cooperative pursuit, whereas the evader is initialized at $(0,0)$. Following the discussion in Section~\ref{sec:imp_dis}, $4$ pursuers are selected to take care of the encirclement condition. Fig.~\ref{fig:RMPC_sim}-\ref{fig:DC_sim} show the snapshots of pursuit processes with our proposed approach, the Voronoi Partition based approach and the ``Direct Charge'' approach, respectively. In all these three figures, the pursuers are marked using red rectangle and the evader is shown using blue circle. In order to illustrate the encirclement formed by the pursuers, we connect the pursuers with brown solid line to form a polygon, which is contained in the true encirclement, in Fig.~\ref{fig:RMPC_sim}, Fig.~\ref{fig:vor_sim}, and Fig.~\ref{fig:DC_sim}. Since our EGP has 4 partitions and at each time step, 4 pursuers act to ensure encirclement and one pursuer use the direct charge policy, we only connect these 4 pursuers in Fig.~\ref{fig:RMPC_sim}.  The green arrows show the direction of the pursuers' move.  As shown in the figure, when the evader is near the boundary of the polygon, the pursuers move in the direction to make sure the evader will be inside the encirclement in the next time step. Results in Fig.~\ref{fig:comparison}-(a) and (b) compare the three methods in terms of the encirclement condition and the capture condition. As shown in our simulation, our TMPC based method can ensure encirclement and capture the evader in the end while the other two methods will fail to keep the evader inside the encirclement before capturing the evader.

\begin{figure}[tp!]
	\centering
	{\includegraphics[width=1.0\linewidth]{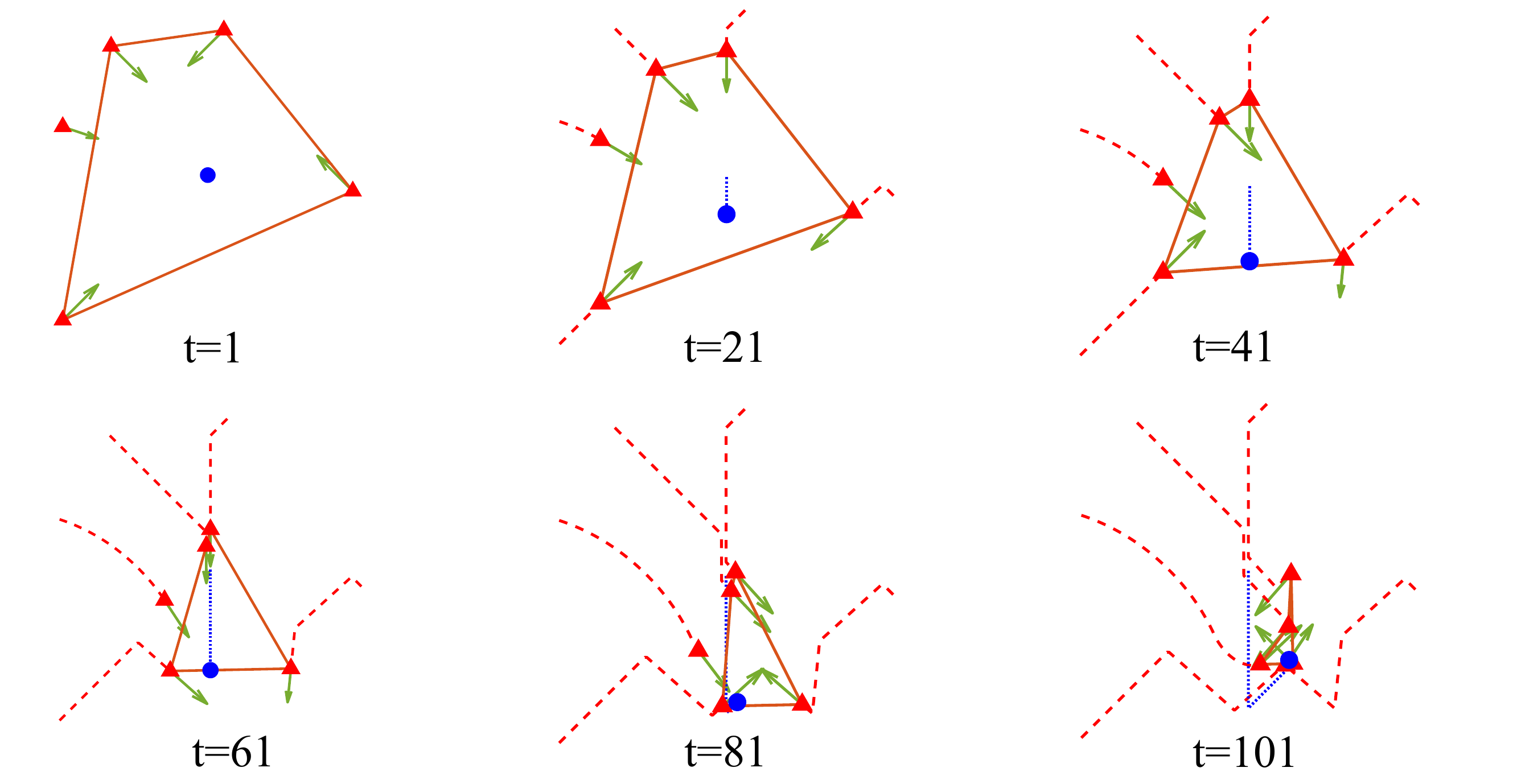}}
	\caption{\footnotesize Simulation snapshots with the proposed TMPC-based method.}
	\label{fig:RMPC_sim}
	\vspace{-15px}
\end{figure}

\begin{figure}[tp!]
	\centering
	{\includegraphics[width=1.0\linewidth]{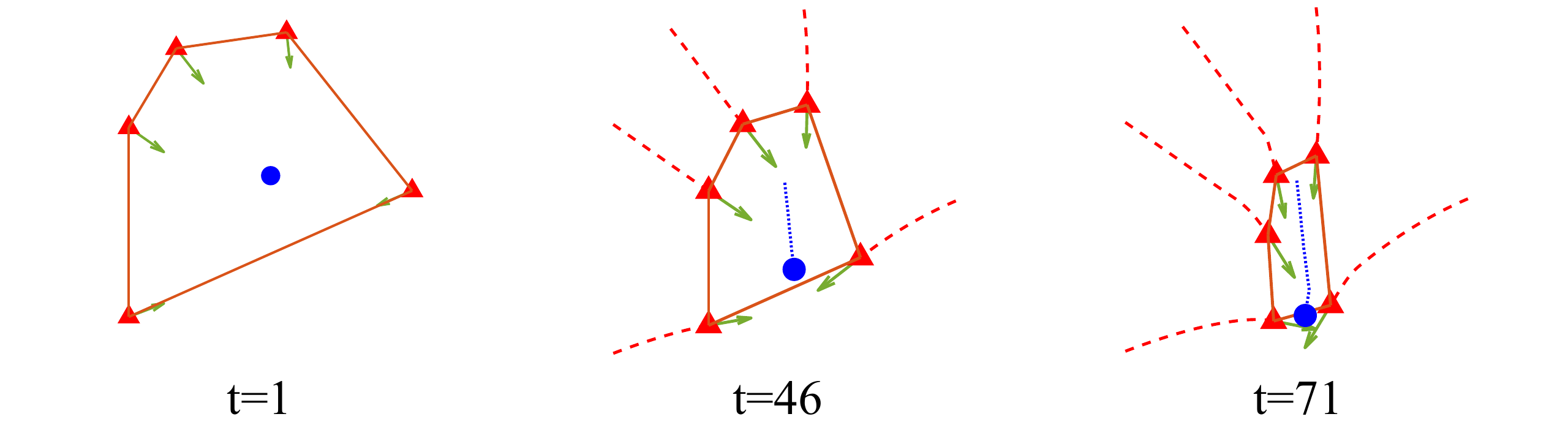}}
	\caption{\footnotesize Simulation snapshots with Voronoi partition method\cite{huang2011guaranteed, zhou2016cooperative}}
	\label{fig:vor_sim}
	\vspace{-10px}
\end{figure}

\begin{figure}[tp!]
	\centering
	{\includegraphics[width=1.0\linewidth]{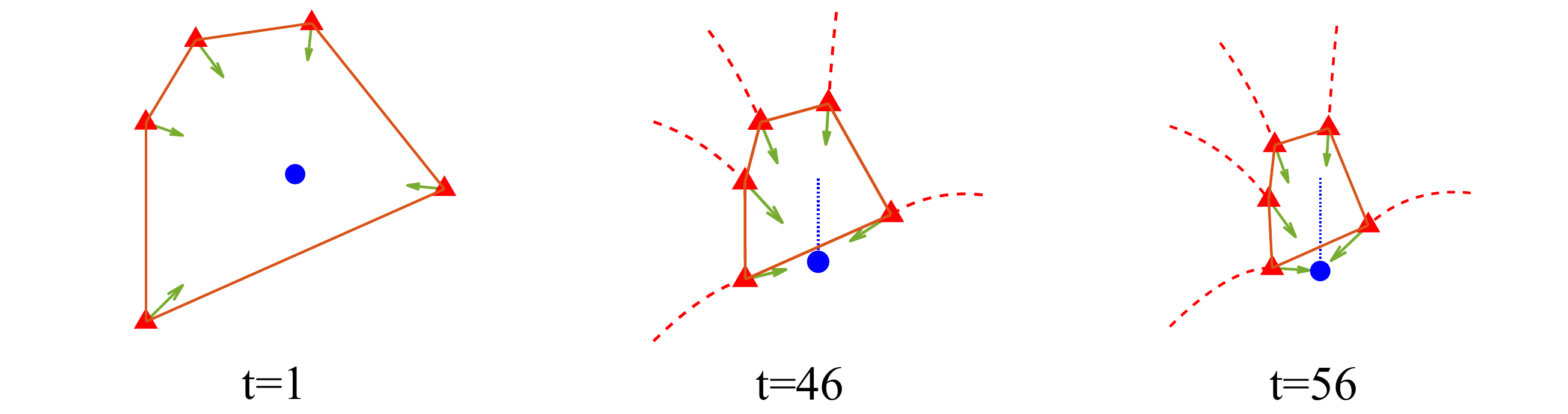}}
	\caption{\footnotesize Simulation snapshots with Direct Charge (DC) method\cite{Mac17RAL}}
	\label{fig:DC_sim}
	\vspace{-15px}
\end{figure}



\begin{figure}[tp!]
	\centering
	{\includegraphics[width=1\linewidth]{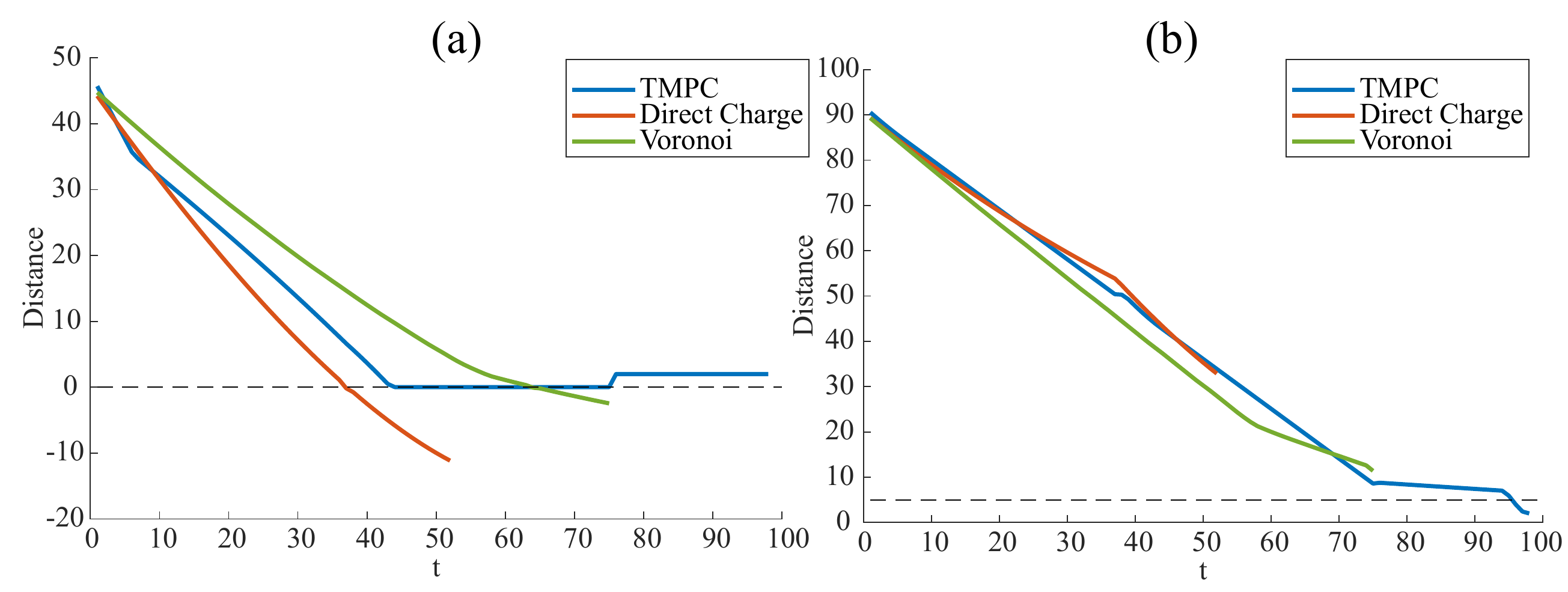}}
	\caption{\footnotesize (a) Distance between the evader and the polygon's  bottom boundary versus time for all three methods. The black dashed line highlighted the zero distance. We regard the distance negative once the evader escape from this boundary. (b) Distance between the evader and the nearest pursuer versus time. The black dashed line denoted the capture radius.}
	\label{fig:comparison}
	\vspace{-15px}
\end{figure}

\section{Conclusions and Future Works}\label{sec:con} 
This paper develops a robust model predictive control (RMPC) based framework and a tractable decentralized solution strategy for the encirclement guaranteed cooperative pursuit problem in an unbounded two-dimensional game domain. By treating the uncertain evader's action as disturbance, the developed RMPC based framework takes care of the encirclement and capture conditions simultaneously. To further address the bilinearities in the RMPC formulation, a novel encirclement guaranteed partitioning scheme is devised, via in-depth geometric analyses of the encirclement condition. Such a partitioning scheme in turn simplifies the original bilinear RMPC problem to a number of linear tube MPC (TMPC) problems efficiently solvable in a decentralized manner.

Possible extensions of this work include consideration of more realistic agent dynamics and multiple evaders. Currently, we require at least 4 pursuers. In the future, we will extend our method to consider the case with only 3 pursuers. On the other hand, validating the proposed solution with hardware experiments is on the list of future works as well.

\bibliographystyle{ieeetr}
\bibliography{mybibfile}

\end{document}

%% file: packages.tex
\usepackage{times,multirow,caption,float}
\usepackage{wrapfig}
\usepackage{graphicx}
\usepackage{microtype}
\usepackage{svg}
\usepackage{epsfig,xspace,layout}
\usepackage[ruled, vlined]{algorithm2e}
\usepackage{subfigure}
\usepackage{color}
\usepackage{amsfonts}
\usepackage{times}
\usepackage{amssymb}
\usepackage{amsmath}
\usepackage{rotating}

\usepackage{amsthm}
\usepackage{graphicx}
\usepackage{epsfig}
\usepackage{mathrsfs}
\usepackage{caption}
\usepackage{subfigure}
\usepackage{mathtools}
\usepackage{makeidx}
\usepackage{multirow} 
\usepackage{dblfloatfix}
\usepackage{threeparttable}
\usepackage{dsfont}
\usepackage[font={small}]{caption}
\usepackage{cleveref}
\usepackage{tikz}
\usetikzlibrary{shapes,arrows,positioning,calc}
\usepackage{siunitx}
\usepackage{url}
\usepackage{cite}

%% file: newcommands.tex
\theoremstyle{definition}
\newtheorem{theorem}{Theorem}
\newtheorem{corollary}{Corollary}

\newtheorem{definition}{Definition}
\newtheorem{problem}{Problem}
\theoremstyle{remark}
\newtheorem{remark}{Remark}

\DeclareMathAlphabet{\mathpzc}{OT1}{pzc}{m}{it}

\DeclareFontFamily{U}{jkpmia}{}
\DeclareFontShape{U}{jkpmia}{m}{it}{<->s*jkpmia}{}
\DeclareFontShape{U}{jkpmia}{bx}{it}{<->s*jkpbmia}{}
\DeclareMathAlphabet{\mathfrak}{U}{jkpmia}{m}{it}

\DeclareMathOperator*{\Min}{Minimize }

\newcommand{\norm}[1]{\left\lVert#1\right\rVert}

\newcommand{\x}{\mathbf{x}}
\newcommand{\w}{\mathbf{w}}
\newcommand{\z}{\mathbf{z}}
\newcommand{\s}{\mathbf{s}}

\newcommand{\0}{\mathbf{0}}
\newcommand{\p}{\mathbf{p}}%
\newcommand{\e}{{\mathbf{e}}}%
\newcommand{\R}{\mathbb{R}}

\newcommand{\NN}{\mathbb{N}}
\renewcommand{\S}{\mathcal{S}}

\newcommand{\X}{\mathcal{X}}

\newcommand{\U}{\mathcal{U}}
\newcommand{\W}{\mathcal{W}}

\newcommand{\E}{\mathcal{E}}

\renewcommand{\P}{\mathcal{P}}

\newcommand{\Q}{\mathcal{Q}}

\newcommand{\al}[1]{\begin{align}#1\end{align}}
\newcommand{\eq}[1]{\begin{equation}#1\end{equation}}
\newcommand{\ald}[1]{\begin{aligned}#1\end{aligned}}

\newcommand{\eqn}[1]{\begin{equation*}#1\end{equation*}}
\newcommand{\aln}[1]{\begin{align*}#1\end{align*}}
\newcommand{\subeq}[1]{\begin{subequations}#1\end{subequations}}
\newcommand{\st}{\text{s.t. }}

%% file: main.bbl
\begin{thebibliography}{10}

\bibitem{Issac1965}
R.~Isaacs, {\em Differential Games}.
\newblock Wiley Press, 1965.

\bibitem{littlewood1986}
J.~E. Littlewood, {\em Littlewood's miscellany}.
\newblock Cambridge University Press, 1986.

\bibitem{NOWAKOWSKI1983235}
R.~Nowakowski and P.~Winkler, ``Vertex-to-vertex pursuit in a graph,'' {\em
  Discrete Mathematics}, vol.~43, no.~2, pp.~235--239, 1983.

\bibitem{AIGNER19841}
M.~Aigner and M.~Fromme, ``A game of cops and robbers,'' {\em Discrete Applied
  Mathematics}, vol.~8, no.~1, pp.~1--12, 1984.

\bibitem{chung2011search}
T.~H. Chung, G.~A. Hollinger, and V.~Isler, ``Search and pursuit-evasion in
  mobile robotics,'' {\em Autonomous robots}, vol.~31, no.~4, pp.~299--316,
  2011.

\bibitem{Croft1964}
H.~T. Croft, ``“lion and man”: A postscript,'' {\em Journal of the London
  Mathematical Society}, vol.~s1-39, no.~1, pp.~385--390, 1964.

\bibitem{V1978}
V.~Janković, ``About a man and lions,'' {\em Matemati\v{c}ki Vesnik},
  vol.~2(15)(30), no.~64, pp.~359--362, 1978.

\bibitem{ibragimov1998game}
G.~Ibragimov, ``A game of optimal pursuit of one object by several,'' {\em
  Journal of applied mathematics and mechanics}, vol.~62, no.~2, pp.~187--192,
  1998.

\bibitem{alexander2009capture}
S.~Alexander, R.~Bishop, and R.~Ghrist, ``Capture pursuit games on unbounded
  domains,'' {\em L’Enseignement Math{\'e}matique}, vol.~55, no.~1,
  pp.~103--125, 2009.

\bibitem{kopparty2005framework}
S.~Kopparty and C.~V. Ravishankar, ``A framework for pursuit evasion games in
  rn,'' {\em Information Processing Letters}, vol.~96, no.~3, pp.~114--122,
  2005.

\bibitem{Isler2005}
V.~{Isler}, S.~{Kannan}, and S.~{Khanna}, ``Randomized pursuit-evasion in a
  polygonal environment,'' {\em IEEE Transactions on Robotics}, vol.~21, no.~5,
  pp.~875--884, 2005.

\bibitem{Deepak2012}
D.~Bhadauria, K.~Klein, V.~Isler, and S.~Suri, ``Capturing an evader in
  polygonal environments with obstacles: The full visibility case,'' {\em The
  International Journal of Robotics Research}, vol.~31, no.~10, pp.~1176--1189,
  2012.

\bibitem{Lavalle2001}
S.~M. {LaValle} and J.~E. {Hinrichsen}, ``Visibility-based pursuit-evasion: the
  case of curved environments,'' {\em IEEE Transactions on Robotics and
  Automation}, vol.~17, no.~2, pp.~196--202, 2001.

\bibitem{Gerkey2006}
B.~P. Gerkey, S.~Thrun, and G.~Gordon, ``Visibility-based pursuit-evasion with
  limited field of view,'' {\em The International Journal of Robotics
  Research}, vol.~25, no.~4, pp.~299--315, 2006.

\bibitem{chen2016multi}
J.~Chen, W.~Zha, Z.~Peng, and D.~Gu, ``Multi-player pursuit--evasion games with
  one superior evader,'' {\em Automatica}, vol.~71, pp.~24--32, 2016.

\bibitem{bacsar1998dynamic}
T.~Ba{\c{s}}ar and G.~J. Olsder, {\em Dynamic noncooperative game theory}.
\newblock SIAM, 1998.

\bibitem{captureFlag}
H.~Huang, J.~Ding, W.~Zhang, and C.~J. Tomlin, ``Automation-assisted
  capture-the-flag: A differential game approach,'' {\em IEEE Transactions on
  Control Systems Technology}, vol.~23, no.~3, pp.~1014--1028, 2015.

\bibitem{huang2011guaranteed}
H.~Huang, W.~Zhang, J.~Ding, D.~M. Stipanovi{\'c}, and C.~J. Tomlin,
  ``Guaranteed decentralized pursuit-evasion in the plane with multiple
  pursuers,'' in {\em 2011 50th IEEE Conference on Decision and Control and
  European Control Conference}, pp.~4835--4840, IEEE, 2011.

\bibitem{zhou2016cooperative}
Z.~Zhou, W.~Zhang, J.~Ding, H.~Huang, D.~M. Stipanovi{\'c}, and C.~J. Tomlin,
  ``Cooperative pursuit with voronoi partitions,'' {\em Automatica}, vol.~72,
  pp.~64--72, 2016.

\bibitem{Mac17RAL}
A.~{Pierson}, Z.~{Wang}, and M.~{Schwager}, ``Intercepting rogue robots: An
  algorithm for capturing multiple evaders with multiple pursuers,'' {\em IEEE
  Robotics and Automation Letters}, vol.~2, no.~2, pp.~530--537, 2017.

\bibitem{SHAH2019103246}
K.~Shah and M.~Schwager, ``Grape: Geometric risk-aware pursuit-evasion,'' {\em
  Robotics and Autonomous Systems}, vol.~121, p.~103246, 2019.

\bibitem{MAYNE2005219}
D.~Mayne, M.~Seron, and S.~Raković, ``Robust model predictive control of
  constrained linear systems with bounded disturbances,'' {\em Automatica},
  vol.~41, no.~2, pp.~219--224, 2005.

\bibitem{Sasa2005}
S.~V. {Rakovic}, E.~C. {Kerrigan}, K.~I. {Kouramas}, and D.~Q. {Mayne},
  ``Invariant approximations of the minimal robust positively invariant set,''
  {\em IEEE Transactions on Automatic Control}, vol.~50, no.~3, pp.~406--410,
  2005.

\bibitem{MPT3}
M.~Herceg, M.~Kvasnica, C.~Jones, and M.~Morari, ``{Multi-Parametric Toolbox
  3.0},'' in {\em Proc.~of the European Control Conference}, (Z\"urich,
  Switzerland), pp.~502--510, July 17--19 2013.
\newblock \url{http://control.ee.ethz.ch/~mpt}.

\end{thebibliography}
